\documentclass[twocolumn,prl,superscriptaddress]{revtex4-1}
\setcitestyle{super}
\usepackage{color,graphicx,subfigure,hyperref,amsthm,amsmath,amssymb}
\usepackage{multirow}
\usepackage{mathtools}
\usepackage{rotating}
\usepackage{url}
\hypersetup{colorlinks,breaklinks,
citecolor=[rgb]{0.0,0.5,0.5},
    urlcolor=[rgb]{0.0,0.5,0.5},
    linkcolor=[rgb]{0.0,0.5,0.5}}
\usepackage{tikz}
\usetikzlibrary{shapes,arrows}

\newtheorem{theorem}{Theorem}
\newtheorem{lemma}[theorem]{Lemma}

\definecolor{nblue}{rgb}{0.2,0.2,0.7}
\definecolor{ngreen}{rgb}{0.1,0.5,0.1}
\definecolor{nred}{rgb}{0.8,0.2,0.2}
\definecolor{nblack}{rgb}{0,0,0}

\renewcommand{\vec}[1]{\mathbf{#1}}

\DeclareMathOperator{\Tr}{Tr}
\def\ket #1{\vert #1\rangle}
\def\bra #1{\langle #1\vert}

\newcommand{\out}[1]{\ensuremath{\ket{#1}\!\bra{#1}}}
\newcommand{\inner}[2]{\ensuremath{\langle {#1}\ket{#2}}}

\newcommand{\Sigmamax}[1]{\langle \Sigma_{\textnormal{tot}}}

\newcommand{\Jam}{Jamio\l kowski }

\newcommand{\Zp}{\ensuremath{\mathbb{Z}_p}}
\newcommand{\nt}[1]{^{\otimes #1}}
\newcommand{\mc}[1]{\mathcal{#1}}
\newcommand{\mbb}[1]{\mathbb{#1}}
\newcommand{\mbf}[1]{\mathbf{#1}}

\newcommand{\hide}[1]{}

\newcommand{\ag}{\ensuremath{\alpha,\gamma}}

\newcommand{\sep}{\ensuremath{I_{\mathrm{sep}}}}
\newcommand{\ent}{\ensuremath{I_{\mathrm{ent}}}}
\newcommand{\SL}[1]{\mathrm{SL}(2,\mbb{Z}_{#1})}
\newcommand{\mat}[2]{\ensuremath{\left(\begin{array}{#1} #2
\end{array}\right)}}

\begin{document}
\title{Contextuality supplies the magic for quantum computation}

\author{Mark Howard}
\email{mark.howard@uwaterloo.ca}
\affiliation{%
Department of Mathematical Physics, National University of Ireland, Maynooth, Ireland}
\affiliation{%
Institute for Quantum Computing and Department of Applied Mathematics,
University of Waterloo, Waterloo, Ontario, Canada, N2L 3G1
}
\author{Joel Wallman}\email{joel.wallman@uwaterloo.ca}
\affiliation{%
Institute for Quantum Computing and Department of Applied Mathematics,
University of Waterloo, Waterloo, Ontario, Canada, N2L 3G1
}%
\author{Victor Veitch}\email{vveitch@uwaterloo.ca}
\affiliation{%
Institute for Quantum Computing and Department of Applied Mathematics,
University of Waterloo, Waterloo, Ontario, Canada, N2L 3G1
}
\affiliation{%
Department of Statistics,
University of Toronto, 100 St. George St. Toronto, Ontario, Canada M5S 3G3
}

\author{Joseph Emerson}\email{jemerson@math.uwaterloo.ca}
\affiliation{%
Institute for Quantum Computing and Department of Applied Mathematics,
University of Waterloo, Waterloo, Ontario, Canada, N2L 3G1
}%
\date{\today}

\begin{abstract}
Quantum computers promise dramatic advantages over their classical counterparts, but the answer to the most basic question ``What is the source of the power in quantum computing?" has remained elusive. Here we prove a remarkable equivalence between the onset of contextuality and the possibility of universal quantum computation via magic state distillation. This is a conceptually satisfying link because contextuality provides one of the fundamental characterizations of uniquely quantum phenomena and, moreover, magic state distillation is the leading model for experimentally realizing fault-tolerant quantum computation. Furthermore, this connection suggests a unifying paradigm for the resources of quantum information: the nonlocality of quantum theory is a particular kind of contextuality and nonlocality is already known to be a critical resource for achieving advantages with quantum communication. In addition to clarifying these fundamental issues, this work advances the resource framework for quantum computation, which has a number of practical applications, such as characterizing the efficiency and trade-offs between distinct theoretical and experimental schemes for achieving robust quantum computation and bounding the overhead cost for the classical simulation of quantum algorithms.
\end{abstract}
\maketitle

Quantum information enables dramatic new advantages for computation, such as Shor's factoring algorithm \cite{Shor:1994} and quantum simulation algorithms~\cite{Lloyd:1996}. This naturally raises the fundamental question: what unique resources of the quantum world enable the advantages of quantum information? There have been many attempts to answer this question, with proposals including the hypothetical ``quantum parallelism''~\cite{Deutsch:1985} some associate with quantum superposition, the necessity of large amounts of entanglement~\cite{Vidal:2003}, and much ado about quantum discord~\cite{Datta:2008}. Unfortunately none of these proposals have proven satisfactory~\cite{Steane:arxiv2000,vDN:2013,anti-discord-paper1,anti-discord-paper2}, and, in particular, none have helped resolve outstanding challenges confronting the field. For example, on the theoretical side, the most general classes of problems for which quantum algorithms might offer an exponential speed-up over classical algorithms are poorly understood. On the experimental side, there remain significant challenges to designing robust, large-scale quantum computers, and an important open problem is to determine the minimal physical requirements of a useful quantum computer~\cite{DQC1, MBQC}. A framework identifying relevant resources for quantum computation should help clarify these issues, for example, by identifying new simulation schemes for classes of quantum algorithms and by clarifying the trade-offs between the distinct physical requirements for achieving robust quantum computation. Here we establish that quantum contextuality, a generalization of nonlocality identified by Bell~\cite{Bell} and Kochen-Specker~\cite{KS} almost 50 years ago, is a critical resource for quantum speed-up within the leading model for fault-tolerant quantum computation, known as \emph{magic state distillation} (MSD)~\cite{BravyiKitaev:2005,Knill:2005,Campbell:2012}. 

Contextuality was first recognized as an intrinsic feature of quantum theory via the Bell-Kochen-Specker ``no-go'' theorem. This theorem implies the impossibility of explaining the statistical predictions of quantum theory in a natural way. In particular, the actual outcome observed under a quantum measurement cannot be understood as simply 
revealing a pre-existing value of some underlying ``hidden variable"~\cite{Mermin:1990}.
A key observation is that the non-locality of quantum theory is a special case of contextuality.
Under the locality restrictions motivating quantum communication, nonlocality is a quantifiable cost for classical simulation  complexity~\cite{Nonlocality} and a fundamental resource for practical applications such as device-independent quantum key distribution~\cite{DIQKD:2007,RUV:2013,VV:arxiv2012}. 
Locality restrictions can be made relevant to measurement-based quantum computation~\cite{MBQC}, for which nonlocality quantifies the resources required to evaluate nonlinear functions~\cite{Raussendorf:2013,HWB:2011}. However, locality restrictions are not relevant in the standard quantum circuit model for quantum computation,  and, in this context, a large amount of entanglement has been shown to be neither necessary nor sufficient for an exponential computational speed-up~\cite{vDN:2013}.

  Here we consider the framework of fault-tolerant (FT) stabilizer quantum computation (QC)~\cite{FTSTABQC} which provides the most promising route to achieving robust universal quantum computation thanks to the discovery of high-threshold codes in $2D$ geometries \cite{RHG:2006,Dennis:2002,Anwar:arxiv2014} (or see, e.g., a review article \cite{Fowler:2012}). In this framework, only a subset of quantum operations---namely, stabilizer operations---can be achieved via a fault-tolerant encoding. These operations define a closed subtheory of quantum theory, the stabilizer subtheory, which is not universal and in fact admits an efficient classical simulation~\cite{Aaronson:2004}. The stabilizer subtheory can be promoted to universal QC through a process known as magic state distillation~\cite{BravyiKitaev:2005,Knill:2005,Campbell:2012} which relies on a large number of ancillary resource states. 
Here we show that quantum contextuality plays a critical role in characterizing the suitability of quantum states for magic state distillation. Our approach builds on the recent work of Cabello, Severini and Winter (CSW)~\cite{CSW:arxiv2010,CSW:2014} that has established a remarkable connection between contextuality and graph-theory. 
We use the CSW framework to identify noncontextuality inequalities such that the onset of state-dependent contextuality, using stabilizer measurements, coincides exactly with the possibility of universal quantum computing via magic state distillation.
The scope of our results differs depending on whether we consider a model of computation using systems of even prime dimension (i.e.~qubits) or odd prime dimension (i.e.~qudits). Whereas in both cases we can prove that violating a non-contextuality inequality is necessary for quantum-computational speed-up via MSD, in the qudit case we are able to prove that a state violates a noncontextuality inequality if and only if it lies outside of the known boundary for MSD.

\textbf{Graph-based contextuality.}---
Interpreting measurements on a quantum state as merely revealing a pre-existing property of the system leads to disagreement with the predictions of quantum theory. In quantum mechanics, a projective measurement can be decomposed as a set of binary tests. Contradictions with models using pre-existing value assignments can arise when these tests appear in multiple measurement scenarios -- i.e., in multiple \emph{contexts}. In other words, we cannot always assign a definite value to tests appearing in multiple contexts and consequently quantum mechanics cannot be described by a noncontextual hidden variable (NCHV) theory. The earliest demonstrations of quantum contextuality used sets of tests such that no NCHV model could reproduce the quantum predictions, regardless of what quantum state was actually measured. Recently, a more general framework has been derived in which a given set of tests can be considered to have noncontextual value assignments only if the measured state satisfies a noncontextuality inequality \cite{CSW:arxiv2010}. We briefly review this framework below.

Consider a set of $n$ binary tests, which can be represented in quantum mechanics by a set of $n$ rank-1 projectors $\{\Pi_1,\ldots,\Pi_n\}$. Two such tests are compatible, and so can be simultaneously performed on a quantum system, if and only if the projectors are orthogonal. We define the witness operator $\Sigma$ for a set of tests to be
\begin{align}
\Sigma=\sum_{i=1}^{n} \Pi_i , \label{eqn:sigmadef}
\end{align}
and the associated exclusivity graph $\Gamma$ to be a graph wherein each vertex corresponds to a projector and two vertices are adjacent (connected) if the corresponding projectors are compatible. Only one outcome can occur when a measurement of a set of  orthogonal projectors is performed, so we require that a value of 1 will be assigned to at most one  projector in each measurement. Since two vertices of $\Gamma$ are adjacent if and only if the corresponding projectors are compatible, the maximum value of $\Sigma$ in an NCHV model, $\langle\Sigma \rangle^{\textsc{NCHV}}_{\max}$, is the independence number $\alpha(\Gamma)$, i.e., the size of the largest set of vertices of $\Gamma$ such that no two elements of the set are adjacent.

The maximum quantum mechanical value of $\Sigma$ can be obtained by varying over projectors satisfying the appropriate compatibility relations and over quantum states. This quantity is bound above by the Lovasz $\vartheta$ number of the exclusivity graph i.e.,
\begin{align}
\langle\Sigma \rangle^{\textsc{QM}}_{\max}\leq \vartheta(\Gamma),
\end{align}
where $\vartheta$ can be calculated as the solution to a semi-definite program. Graphs for which $\protect{\alpha(\Gamma) < \vartheta(\Gamma)}$ indicate that appropriately chosen projectors $\{\Pi_i\}$ and states $\rho$ may reveal quantum contextuality by violating the noncontextuality inequality
\begin{align}\label{eqn:witness}
\Tr (\Sigma \rho) \leq \alpha(\Gamma)	\,.
\end{align}
For generalized probabilistic theories (GPT), an important class of ``post-quantum'' theories, the maximum value of $\Sigma$ is given by the fractional packing number of the exclusivity graph $\alpha^*(\Gamma)$ i.e.,
\begin{align}
\langle\Sigma \rangle^{\textsc{GPT}}_{\max}=\alpha^*(\Gamma).
\end{align}
Note that if $\alpha(\Gamma) <\langle\Sigma \rangle^{\textsc{QM}}_{\max} =\alpha^*(\Gamma)$, then the optimal choice of quantum state and projectors is maximally contextual, in that no greater violation of the noncontextuality inequality can be obtained in any GPT.

\textbf{The stabilizer formalism.}---
Quantum information theory relies heavily on a family of finite groups usually called the (generalized) Pauli groups. The most promising and well understood quantum error correcting codes---stabilizer codes---are built using the elements of these groups, i.e., Pauli operators. Two-level quantum systems---qubits---are the most commonly used building blocks for quantum computing, but a circuit using $d$-level systems---qudits---has the same computational power. While qudits with larger values of $d$ may pose new experimental challenges, these may be offset by a lower overhead for fault-tolerant computation \cite{Campbell:2012}. In this subsection we outline the mathematical structure associated with the generalized Pauli group and the geometrical characterization of probabilistic mixtures of stabilizer states.

The stabilizer formalism for $p$-dimensonal systems ($p$ a prime number) is defined using the generalized $X$ and $Z$ operators
\begin{align}
X\ket{j} =\ket{j+1} \quad Z\ket{j}=\omega^j\ket{j}\,,
\end{align}
where $\omega = \exp(\tfrac{2\pi i}{p})$. The set of Weyl-Heisenberg displacement operators is defined as
\begin{align}
\mbf{D}_p = \{D_{x,z}=\omega^{2^{-1}xz} X^x Z^z:x,z\in\mbb{Z}_p\}	\,,\label{eqn:DispOps}
\end{align}
where $2^{-1}$ is the multiplicative inverse of $2$ in the finite field $\mbb{Z}_p=\{0,1,\ldots,p-1\}$. For $p=2$, one can replace $\omega^{-2^{-1}} $ with $i$ in Eq.~\eqref{eqn:DispOps} to recover the familiar qubit Pauli operators. The Clifford group $\mbf{C}_{p,n}$ is defined to be the normalizer of the group $\langle\mbf{D}_p\nt{n}\rangle$ (i.e., the group generated by the set of displacement operators), that is,
\begin{align}
\mbf{C}_{p,n} = \{U\in\mc{U}(d^n): U \langle\mbf{D}_p\nt{n}\rangle U^\dag =\langle\mbf{D}_p\nt{n}\rangle \}	\,,
\end{align}
and the set of stabilizer states is the image of the computational basis under the Clifford group $\mbf{C}_{p,n}$.

The stabilizer polytope is the convex hull of the set of stabilizer states. For a single system, the stabilizer polytope \cite{Cormick:2006} is defined by the following set of simultaneous inequalities
\begin{align}\label{eqn:Pstab}
\mathcal{P}_{\textsc{STAB}}=\left\{\rho\, :\Tr(\rho A^{\vec{q}})\geq 0,\, \vec{q} \in \mathbb{Z}_p^{p+1} \right\}\,
\end{align}
where $A^{\vec{q}}= -\mathbb{I}_p + \sum_{j=1}^{p+1} \Pi_j^{q_j}$ and
 $\Pi_j^{q_j}$ is the projector onto the eigenvector with eigenvalue $\omega^{q_j}$ of the $j$th operator in the list
$\left\{D_{0,1},D_{1,0},D_{1,1},\ldots,D_{1,p-1}\right\}$ [the eigenbases of these operators form a complete set of mutually unbiased bases (MUBs)].

\textbf{Magic state distillation.}---The stabilizer formalism of the previous subsection was developed in the search for quantum error-correcting codes, i.e., codes allowing the robust, fault-tolerant storage and manipulation of quantum information stored across many subsystems \cite{FTSTABQC,Knill:2005}. Surface codes \cite{Kitaev:2003,Dennis:2002,RHG:2006,Fowler:2012,Anwar:arxiv2014}, in particular, admit a comparatively high fault-tolerance threshold within an experimentally realistic planar physical layout. Codes such as these have a  finite non-universal set of transversal (i.e., manifestly fault-tolerant) operations that must be supplemented with an additional resource -- a supply of so-called magic states -- in order to attain universality. Magic state distillation (MSD) refers to the subroutine, described below, wherein almost pure resource states are constructed using large numbers of impure resource states~\cite{BravyiKitaev:2005,Knill:2005,Campbell:2012}.

An MSD protocol consists of the following steps: (i) Prepare $n$ copies of a suitable (see later) input state, i.e., $\rho_{\text{in}}^{\otimes n}$ (ii) Perform a Clifford operation on $\rho_{\text{in}}^{\otimes n}$ (iii) Perform a stabilizer measurement on all but the first $m$ registers, postselecting on a desired outcome. With appropriate choices of stabilizer operations, the resulting output state in the first $m$ registers, $\rho_{\text{out}}^{\otimes m}$, is purified in the direction of a magic state $\ket{\nu}$, so that $\langle{\nu}|\rho_{\text{out}}\ket{\nu} > \langle{\nu}|\rho_{\text{in}}\ket{\nu}$. This process can be reiterated until $\rho_{\text{out}}$ is sufficiently pure, at which point the resource $\rho_{\text{out}}$ is used up to approximate a non-Clifford operation (via ``state injection''), e.g., the $\pi/8$ gate or its qudit generalizations \cite{Campbell:2012,HV:2012}. Supplementing stabilizer operations with the ability to perform such gates enables fault-tolerant and universal QC.

\begin{figure}[ht]
\centering
\includegraphics[scale=0.5]{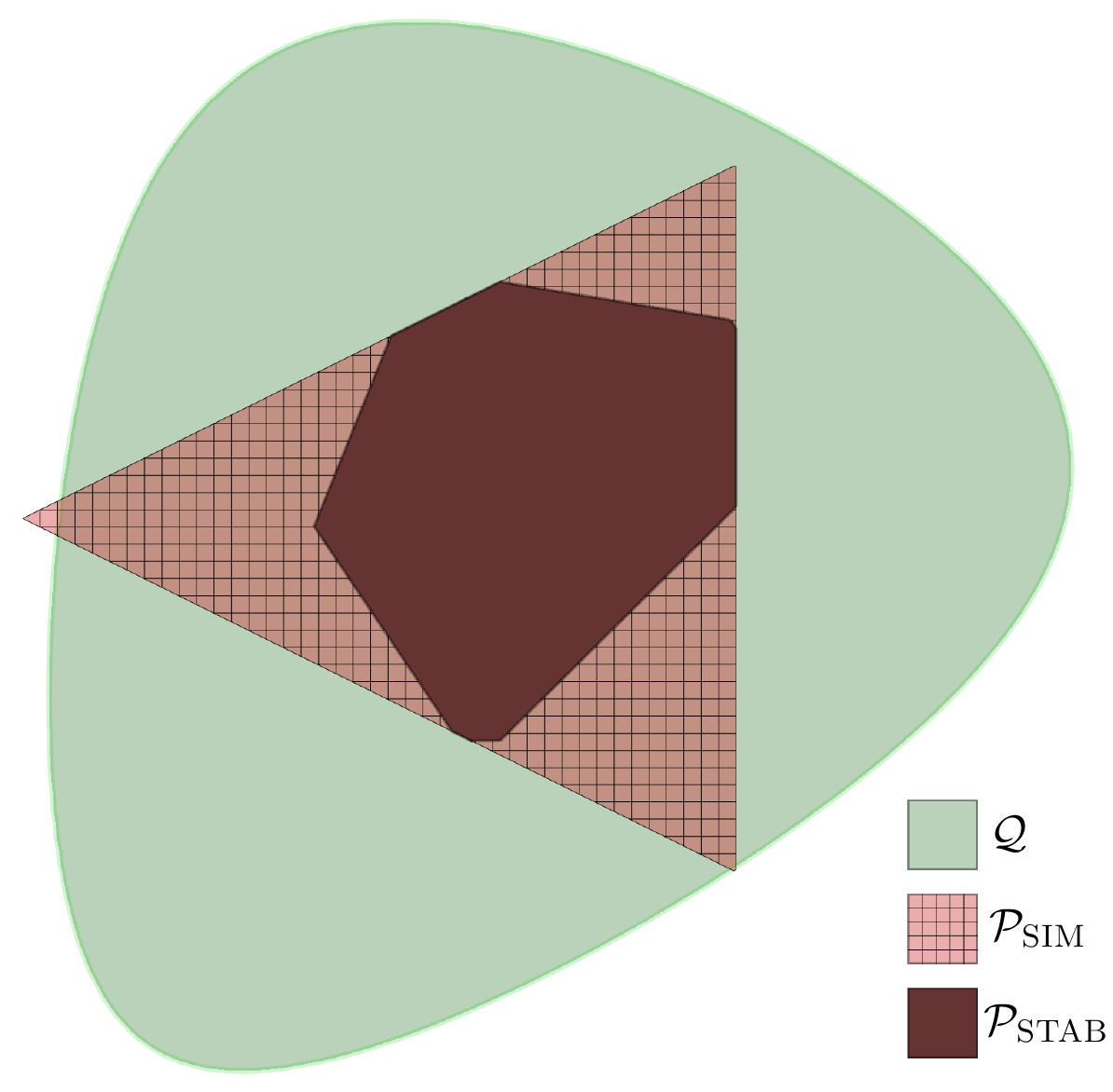}
\caption{\label{fig:slice}\textbf{A $2$-dimensional slice through qutrit state space.} Three distinct regions in the space of Hermitean operators: Region $\mathcal{Q}$, shaded in pale green, describing quantum state space (density operators), region $\mathcal{P}_{\textsc{SIM}}$, with hatched shading, corresponding to ancillas known to be efficiently simulable (and hence useless for quantum computation via Magic State Distillation) and the dark red region $\mathcal{P}_{\textsc{STAB}}$ describing mixtures of stabilizer states; the strict inclusion $\mathcal{P}_{\textsc{STAB}} \subset \mathcal{Q} \bigcap \mathcal{P}_{\textsc{SIM}}$ identifies a large class of \emph{bound magic states} \cite{Veitch:2012}.
}
\end{figure}

For which states $\rho_{\text{in}}$ does there exist an MSD routine purifying $\rho_{\text{out}}$ towards a non-stabilizer state? A large subset of quantum states have been ruled out by virtue of the fact that efficient classical simulation schemes are known for noiseless stabilizer circuits supplemented by access to an arbitrary number of states from the polytope $\rho_{\text{in}}\in \mathcal{P}_{\textsc{SIM}}$ ~\cite{Aaronson:2004,Veitch:2012,Mari:2012}. This polytope $\mathcal{P}_{\textsc{SIM}}$ of the known simulable states is prescribed by \cite{Cormick:2006,WvDMH:2010}
\begin{align}\label{eqn:Psim}
\mathcal{P}_{\textsc{SIM}}=\begin{cases} \left\{\rho\, :\Tr(\rho A^{\vec{r}})\geq 0,\,\vec{r} \in \mathbb{Z}^3_2 \right\}\, &p=2,\\
\left\{\rho\, :\Tr(\rho A^{x\vec{a}+z\vec{b}})\geq 0,\, x,z \in \Zp \right\}\, &p>2 \end{cases}
\end{align}
where $\vec{a}=[1,0,1,\ldots,p-1]$ and $\vec{b}=-[0,1,1,\ldots,1]$ \cite{Appleby:2008}. Note that $\mathcal{P}_{\textsc{SIM}}=\mathcal{P}_{\textsc{STAB}}$ for qubits (giving an octahedron inscribed within the Bloch sphere) whereas $\mathcal{P}_{\textsc{SIM}}\supset \mathcal{P}_{\textsc{STAB}}$ is a proper superset for all other primes. Subsequently we refer to the set of facets enclosing $\mathcal{P}_{\textsc{SIM}}$ as
\begin{align}
\mc{A}_{\textsc{SIM}}=\{A^{\vec{r}} | p=2:\vec{r} \in \mathbb{Z}^3_2, p\neq2:\vec{r}= x\vec{a}+z\vec{b}\}	\,.
\end{align}
In Fig.~1 we plot the geometric relationship between arbitrary quantum states, and sets of states contained within $\mathcal{P}_{\textsc{SIM}}$ and $\mathcal{P}_{\textsc{STAB}}$ for the case of qutrits ($p=3$).

By prior results \cite{Wallman:2012, Veitch:2012}, we know that the set of states $\mathcal{P}_{\textsc{SIM}}$ coincides exactly with the set of states that are nonnegatively represented within a distinguished quasiprobability representation---a discrete Wigner function (DWF) \cite{Wootters:DWF1,Wootters:DWF2,Gross:2006}. Are the states in the set $\mathcal{P}_{\textsc{SIM}}$, the set excluded from MSD by the known efficient simulation schemes, the complete set of non-distillable states? We now address this fundamental question by demonstrating a remarkable relationship between non-distillability, non-negativity and non-contextuality.  

\begin{figure}[b]
\centering
\includegraphics[scale=0.95]{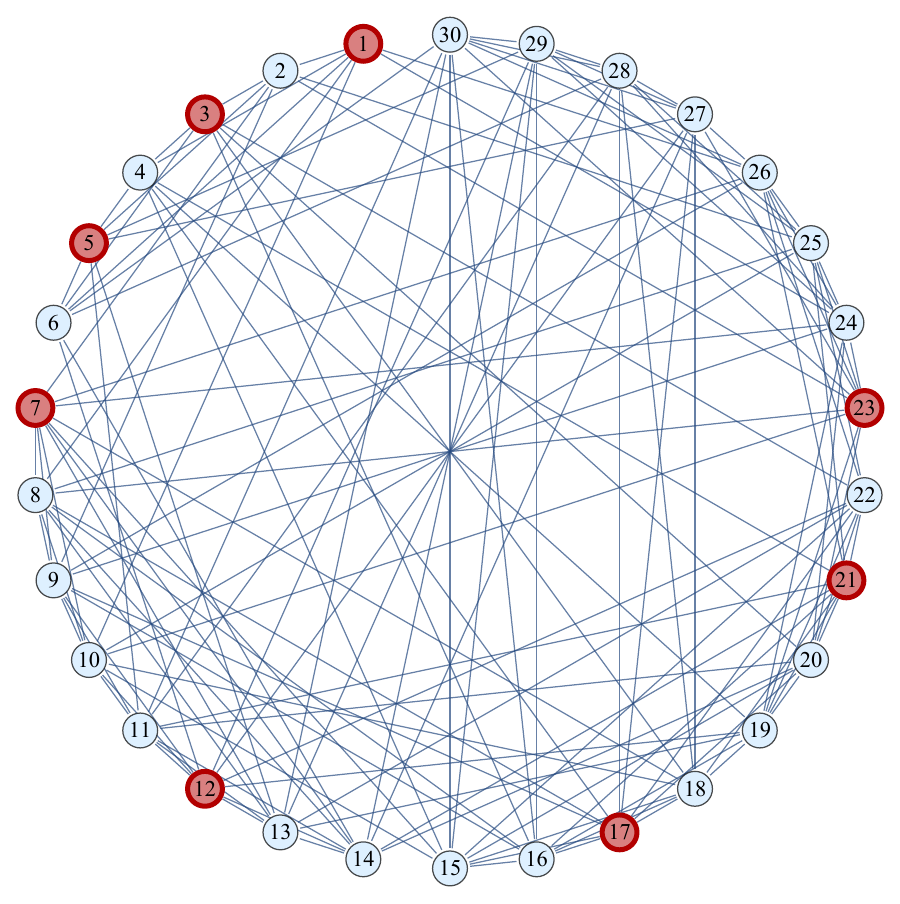}
\caption{\label{fig:graph} \textbf{ Our construction applied to $2$ qubits.} Each of the $30$ vertices in this graph $\Gamma$ corresponds to a $2$-qubit stabilizer state; connected vertices correspond to orthogonal states. A maximum independent set (representing mutually non-orthogonal states) of size $\alpha(\Gamma)=8$ is highlighted in red. As described in Thm.~\ref{thm:maintheorem}, this value of $\alpha$ identifies all states $\rho \notin \mathcal{P}_{\textsc{SIM}}$ as exhibiting contextuality with respect to the stabilizer measurements in our construction.}
\end{figure}

\textbf{Contextuality as a computational resource}---We will prove that all states $\rho \notin \mathcal{P}_{\textsc{SIM}}$ exhibit state-dependent contextuality with respect to stabilizer measurements.  Our definition of stabilizer measurement is quite inclusive; we allow all projective measurements wherein elements are rank-$1$ projectors onto stabilizer states. Rearranging the definition of $A^{\vec{r}}$ given in Eq.~\eqref{eqn:Pstab} gives
\begin{align}\label{eqn:sum_proj}
\sum_{j=1}^{p+1}\sum_{\stackrel{s_j\in\mbb{Z}_p}{s_j\neq r_j}} \Pi_j^{s_j} = p\mbb{I}_p - A^{\vec{r}}	\,,
\end{align}
that is, the set of projectors $\{\Pi_j^{s_j\neq r_j}\}$ is a set of projectors whose sum, $\Sigma^{\vec{r}}$, is such that
\begin{align}
\Tr (\Sigma^{\vec{r}}\rho) > p \Leftrightarrow \Tr(A^{\vec{r}}\rho)<0	.
\end{align}
The left hand side of this equivalence is a witness for contextuality if and only if the independence number of the associated graph $\Gamma^{\vec{r}}$ satisfies $\alpha(\Gamma^{\vec{r}})=p$ as in Eq.~\eqref{eqn:witness}. In fact, this simple construction fails to identify \emph{any} quantum states as contextual because $\Tr (\Sigma\rho) \leq \alpha(\Gamma^{\vec{r}})$ for all $\rho$. This is not surprising given that every single-qudit stabilizer projector is part of exactly one context, namely the basis (one of the complete set of MUBs) in which it is contained.

Stabilizer projectors appear in multiple contexts only when two or more subsystems are involved. Consequently, we introduce two-qudit stabilizer projectors such that the structure of $A^{\vec{r}}$ as in Eq.~\eqref{eqn:sum_proj} is reflected on the first qudit only. We can limit consideration to two-qudit projectors since this approach characterizes as contextual all single-qudit states that do not have an NCHV model via the discrete Wigner function, i.e., we find two-qudit projectors are sufficient to achieve the best possible result.

Our construction uses a different set of projectors for each facet $A^{\vec{r}}$. For a fixed facet $A^{\vec{r}}$, we define a set of separable projectors
\begin{align}
\{\Pi\}^{\vec{r}}_\textrm{sep}=\{\Pi_j^{s_j\neq r_j}\otimes\out{k}:\, 1\leq j\leq p+1,s_j,k\in\mbb{Z}_p\}
\end{align}
that is, we take the $p(p^2-1)$ separable projectors consisting of all tensor products of projectors in Eq.~\eqref{eqn:sum_proj} for the first qudit and computational basis states for the second qudit. We also define the set $\{\Pi\}_\textrm{ent}$ to be the set of all two-qudit entangled projectors.

The sum of the combined set of separable and entangled projectors $\{\Pi\}^{\vec{r}}=\{\Pi\}^{\vec{r}}_\textrm{sep}\cup\{\Pi\}_\textrm{ent}$ is
\begin{align}
\Sigma^{\vec{r}}= (p^3\mbb{I}_p - A^{\vec{r}})\otimes\mbb{I}_{p} \label{eqn:Sigmatot}
\end{align}
so that for any state $\sigma \in \mathcal{H}_p$ of the second system (even the maximally mixed state) we have
\begin{align}
\Tr\left[\Sigma^{\vec{r}}\left(\rho \otimes \sigma\right)\right]\leq p^3 \iff \Tr\left[A^{\vec{r}} \rho \right] \geq 0. \label{eqn:bijection}
\end{align}
Forming the exclusivity graph $\Gamma^{\vec{r}}$ of $\{\Pi\}^{\vec{r}}$ and applying the results of CSW identifies the left hand side of Eq.~\eqref{eqn:bijection} as a witness for the contextuality of $\rho$.
The following theorem shows that the inequality on the left-hand-side of Eq.~\eqref{eqn:bijection} is indeed a noncontextuality inequality.

\begin{theorem}\label{thm:maintheorem}
The independence number of the exclusivity graph associated with $\Sigma^{\vec{r}}$ is $\alpha(\Gamma^{\vec{r}})=
p^3$ for all $A^{\vec{r}}\in \mc{A}_{\textsc{SIM}}$ and all prime $p\geq2$.
Furthermore, for $p>2$, a state exhibits contextuality if and only if it violates one of our noncontextuality inequalities and maximally contextual states saturate the bound on contextuality associated with post-quantum generalized probabilistic theories i.e.,
\begin{align}
\langle \Sigma^{\vec{r}}\rangle^{2-\textsc{qudit}}_{\max} =\vartheta(\Gamma^{\vec{r}})=\alpha^*(\Gamma^{\vec{r}})=p^3+1\ \quad &(p>2).
\end{align}
\end{theorem}

Theorem \ref{thm:maintheorem} says that, relative to our construction, exactly the states $\rho \notin \mathcal{P}_{\textsc{SIM}}$ are those that exhibit contextuality. For qudits of odd prime dimension there does not exist any construction using stabilizer measurements that characterizes any $\rho \in \mathcal{P}_{\textsc{SIM}}$ as contextual, so that the conditions for contextuality and the possibility of quantum speed-up via magic state distillation coincide exactly.

\begin{proof}
For $p=2$, a software package~\cite{Ostergard:2002} can be used to obtain
\begin{align*}
\alpha(\Gamma^{\vec{r}}) = 8 < \vartheta(\Gamma^{\vec{r}}) < \alpha^*(\Gamma^{\vec{r}})\leq 9 \,.
\end{align*} 
The exclusivity graph $\Gamma^{\vec{r}}$ and an independent set of 8 vertices is depicted in Fig.~2. The maximal violation of our noncontextuality inequality is achieved by the state $\out{T}\otimes \sigma $, where $\ket{T}$ is the magic state introduced in \cite{BravyiKitaev:2005} and $\sigma$ is arbitrary.

For $p>2$, we will now show that $\vartheta(\Gamma^{\vec{r}})=\alpha^*(\Gamma^{\vec{r}}) = p^3+1$. To do this, we use the graph theoretical inequality
\begin{align}
\alpha(\Gamma)\leq \langle \Sigma \rangle^{\textsc{QM}}_{\max} \leq \vartheta(\Gamma) \leq \alpha^*(\Gamma) \leq \bar{\chi}(\Gamma), \label{eqn:inequalities}
\end{align}
where $\bar{\chi}(\Gamma) \in \mathbb{N}$ is the clique cover number, which is the minimum number of cliques needed to cover every vertex of $\Gamma$. The clique cover number cannot be greater than the number of distinct bases in $\{\Pi\}^{\vec{r}}$, which contains $p+1$ separable bases and $p^3 - p$ entangled bases. Therefore $\vartheta(\Gamma) \leq \alpha^*(\Gamma) \leq p^3 + 1$. Then, since there exist\cite{Gross:2006} quantum states $\rho$ such that $\Tr(A^{\vec{r}}\rho) = -1$, $\langle \Sigma^{\vec{r}} \rangle^{\textsc{QM}}_{\max} \geq p^3+1$ by Eq.~\eqref{eqn:bijection} and so $\vartheta(\Gamma^{\vec{r}}) = \alpha^*(\Gamma^{\vec{r}}) = p^3 +1$.

The statement that no noncontextuality inequality constructed from stabilizer measurements can be violated by any state $\rho\in\mathcal{P}_{\textsc{SIM}}$ for odd prime dimensions follows from the existence of a NCHV model, namely, the discrete Wigner function~\cite{Gross:2006,Veitch:2012,Wootters:DWF1,Wootters:DWF2}, for all stabilizer measurements and states $\rho\in\mathcal{P}_{\textsc{SIM}}$.

We defer the proof that $\alpha(\Gamma^{\vec{r}}) = p^3$ to the supplementary material.
\end{proof}

\textbf{Significance and outlook.}---
For qudits of odd prime dimension (hereafter referred to simply as qudits), a state is non-contextual under the available set of measurements---stabilizer measurements---if and only if it lies in the polytope $\mathcal{P}_{\textsc{SIM}}$ (the set of ancilla states known to be useless for any magic state distillation routine). The same construction applied to qubits also identifies all $\rho \not\in \mathcal{P}_{\textsc{SIM}}$ as contextual. These results establish that contextuality is a necessary resource for universal quantum computation via MSD.

%
%

For qudits, the set of states proven to be contextual by our construction have been previously conjectured to be sufficient~\cite{Veitch:2014} to promote stabilizer circuits to universality. Proving this conjecture would require proving that any state $\rho \notin\mathcal{P}_{\textsc{SIM}}$ can be distilled to a magic state. While substantial progress in this direction has been made~\cite{Campbell:2012}, it is still an open problem.

For qubits, the mere presence of contextuality cannot be sufficient to promote stabilizer circuits to universality since any state $\rho\in \mathcal{P}_{\textsc{SIM}}$ (which includes the maximally mixed state) can violate a noncontextuality inequality constructed from stabilizer measurements. For example, converting the Peres-Mermin magic square~\cite{Peres:1991,Howard:2013a} to a 24-ray (projector) proof of contextuality and applying the CSW formalism gives a noncontextuality inequality that is violated by all two-qubit states including states of the form $\rho \otimes \sigma=\mathbb{I}/4$.

The crucial difference between qubits and qudits is that state-independent contextuality (like that of the Peres-Mermin square) is never manifested within the qudit stabilizer formalism. Consequently, for qudits, any contextuality is necessarily state-dependent and our results show that this contextuality has an operational meaning as necessary and possibly sufficient for the ``magic'' that makes quantum computers tick. In the qubit case, it is a pressing open question whether a suitable operationally-motivated refinement~\cite{Spekkens:2005,Acin:arxiv2012} or quantification of contextuality can align more precisely with the potential to provide a quantum speed-up.

\bigskip

\noindent \textbf{Acknowledgements}
M.H.~was financially supported by the Irish Research Council (IRC) as part of the Empower Fellowship program, and all authors acknowledge financial support from CIFAR and the Government of Canada through NSERC.
%
%
%

\bigskip

\textbf{Supplementary Information: Contextuality supplies the magic for quantum computation}

Here we prove that, for odd-prime $p$, the independence number of $\Gamma^{\vec{r}}$ is $p^3$. Recall that the independence number of a graph $\Gamma$ is the size of the largest independent set of $\Gamma$, where an independent set is a set of vertices of which no two are connected. Since two vertices are connected if and only if the associated projectors commute, an independent set in $\Gamma^{\vec{r}}$ is equivalent to a set of mutually noncommuting projectors in $\{\Pi\}^{\vec{r}}$. Since the elements of $\{\Pi\}^{\vec{r}}$ are all rank 1, two elements are noncommuting if and only if they are nonorthogonal. 

We prove $\alpha(\Gamma^{\vec{r}}) = p^3$ by proving $\alpha(\Gamma^{\vec{r}}) \geq p^3$ and $\alpha(\Gamma^{\vec{r}}) < p^3 +1$. This completes the proof since $\alpha(\Gamma^{\vec{r}})$ is an integer. In Theorem~\ref{thm:lower} we show that $\alpha(\Gamma^{\vec{r}}) \geq p^3$ by showing that there exists a set of $p^3$ mutually nonorthogonal elements of $\{\Pi\}^{\vec{r}}$ for any $A^{\vec{r}}\in \mc{A}_{\rm sim}$. In Lemmas~\ref{lem:orth}--\ref{lem:entsep} we parametrize the set of stabilizer projectors using the symplectic representation of the Clifford group in order to transform a condition of mutual nonorthogonality of projectors into a set of algebraic constraints on parameters. In Theorem~\ref{thm:upper} we then show that $\alpha(\Gamma^{\vec{r}}) < p^3 + 1$ by showing that no subset of $p^3 + 1$ elements of $\{\Pi\}^{\vec{r}}$ can satisfy the constraints established in Lemmas~\ref{lem:orth}--\ref{lem:entsep}, that is, there cannot exist a subset of $p^3 +1$ mutually nonorthogonal elements of $\{\Pi\}^{\vec{r}}$.

\begin{theorem}\label{thm:lower}
For any $A^{\vec{r}}\in \mc{A}_{\rm sim}$, the independence number of the exclusivity graph $\Gamma^{\vec{r}}$ satisfies $\alpha(\Gamma^{\vec{r}})\geq p^3$.
\end{theorem}

\begin{proof}
To prove $\alpha(\Gamma^{\vec{r}})\geq p^3$, it is sufficient to show that there exists a set of $p^3$ mutually nonorthogonal elements of $\{\Pi\}^{\vec{r}}$. We accomplish this using the phase-space formalism for stabilizer projectors$^{34,40}$.

The phase-space formalism maps stabilizer projectors to value assignments over a phase space $\mbb{Z}_p^4$. For two qudits, the map is given by the discrete Wigner function
\begin{align}
W_{\Pi}(\vec{u},\vec{v}) = \Tr(\Pi A^{\vec{u}}\otimes A^{\vec{v}})
\end{align}
where $A^{\vec{u}},A^{\vec{v}}\in\mc{A}_{\rm sim}$. The discrete Wigner function takes on the values $\{0,1\}$ for stabilizer projectors$^{40}$. Note that the above map is noncontextual since it depends only upon $\Pi$ and not upon which other projector is measured. Since
\begin{align}
\Tr(\Pi \Pi') = p^{-1}\sum_{\vec{u},\vec{v}} W_{\Pi'}(\vec{u},\vec{v})W_{\Pi}(\vec{u},\vec{v})	\ \forall \Pi,\Pi'	\,,
\end{align}
the set of projectors that assign the value 1 to a point $(\vec{u},\vec{v})$ in phase space are mutually nonorthogonal.

By the linearity of the trace, the number of elements of $\{\Pi\}^{\vec{r}}$ that assign the value 1 to $(\vec{u},\vec{v})$ is
\begin{align}
\sum_{\Pi\in \{\Pi\}^{\vec{r}}} W_{\Pi}(\vec{u},\vec{v}) &= \sum_{\Pi\in \{\Pi\}^{\vec{r}}} \Tr(\Pi A^{\vec{u}}\otimes A^{\vec{v}}) \notag\\
&= \Tr\Bigl[\Bigl(\sum_{\Pi\in \{\Pi\}^{\vec{r}}} \Pi\Bigr) A^{\vec{u}}\otimes A^{\vec{v}}\Bigr] \notag\\
&= \Tr\Bigl[\Sigma^{\vec{r}} A^{\vec{u}}\otimes A^{\vec{v}}\Bigr]	\notag\\
&= \Tr\Bigl[\left[(p^3\mbb{I}_p - A^{\vec{r}})\otimes\mbb{I}_{p}\right] A^{\vec{u}}\otimes A^{\vec{v}}\Bigr]	\notag\\
&= p^3 - \delta(\vec{r} - \vec{u})\,,
\end{align}
where the last line follows from $\Tr A^{\vec{u}} = 1$ and $\Tr (A^{\vec{r}} A^{\vec{u}}) = \delta_{\vec{r},\vec{u}}$. Therefore for any $\vec{u}\neq \vec{r}$, exactly $p^3$ elements of $\Pi\in \{\Pi\}^{\vec{r}}$ assign the value 1 to $(\vec{u},\vec{v})$ for any $v$ and so are mutually nonorthogonal.
\end{proof}

To prove that for any $A^{\vec{r}}\in \mc{A}_{\rm sim}$, any set of $p^3 + 1$ mutually nonorthogonal stabilizer projectors necessarily contains elements outside of $\{\Pi\}^{\vec{r}}$, we will parametrize the set of stabilizer projectors in order to transform a condition of mutual nonorthogonality into a set of algebraic constraints on parameters.

The specific parametrization we use is the symplectic representation of the Clifford group$^{38}$. In this parametrization, Clifford elements are written as
\begin{align}
C = D_{x,z} U_F	\,,
\end{align}
where $D_{x,z}$ is as defined in Eq.~(6) of the main text and
\begin{align}
F = \mat{cc}{\alpha & \beta \\ \gamma & \epsilon}
\end{align}
is an element of the symplectic group $\SL{p}$, that is, the entries of $F$ are elements of $\mbb{Z}_p$ and $\det F = 1$, and
\begin{align}\label{eq:UF}
U_F=\begin{cases}
\frac{1}{\sqrt{p}}\sum_{j,k=0}^{p-1}\tau^{\beta^{-1}\left(\alpha k^2-2 j k +\epsilon j^2\right)}\ket{j}\bra{k}\quad &\beta\neq0\\
\sum_{k=0}^{p-1} \tau^{\alpha \gamma k^2} \ket{\alpha k}\bra{k}\quad &\beta=0
\end{cases}
\end{align}
where $\tau=\omega^{2^{-1}}$.

In what follows, we will treat the Pauli and the symplectic components separately. For Pauli operators we have $D_{x,z}^\dag \propto D_{-x,-z}$, while for symplectic gates $U_F^\dag = U_{F^{-1}}$. We then have
\begin{align}
U_F D_{x,z} U_F^\dag = D_{F(x,z)} \,.
\end{align}
An important feature of the two-qudit Clifford group that enables the following proof is that the set of two-qudit entangled stabilizer states is exactly the set
\begin{align}
\ket{x,z,F} := (D_{x,z}U_F\otimes \mbb{I})\ket{\Phi} \qquad (\ket{\Phi}=\sum_j\ket{jj}/\sqrt{p})
\end{align}
of states \Jam isomorphic to the single-qudit Clifford group. Moreover, $F$ labels an orthonormal entangled basis, while the Pauli component selects an element of the basis$^{49}$.

The following Lemma, proven elsewhere$^{49}$, provides conditions for two entangled states to be non-orthogonal.

\begin{lemma}\label{lem:orth}
Two entangled stabilizer states $\ket{x_1,z_1,F_1}$ and $\ket{x_2,z_2,F_2}$ are non-orthogonal if and only if $\Tr F =2$ and
\begin{align}
\beta_F\Delta z = (1-\alpha_F)\Delta x &\quad (\beta_F\neq 0) \notag\\
\Delta x = 0 &\quad (\beta_F = 0)
\end{align}
where
\begin{align}
F=F_1^{-1}F_2 = \mat{cc}{\alpha_F & \beta_F \\ \gamma_F & \epsilon_F}
\end{align}
and
\begin{align}
\mat{c}{\Delta x\\ \Delta z} = F_1^{-1}\mat{c}{x_2 - x_1\\ z_2 - z_1}	\,.
\end{align}
\end{lemma}

To determine when entangled and product stabilizer states are nonorthogonal, we define the group of computational-basis-preserving gates contained within $\SL{p}$ to be
\begin{align}
BP:=\left\{C_{\ag}=\mat{cc}{\alpha & 0 \\\gamma & \alpha^{-1}}: \alpha\in\Zp^*, \gamma \in \Zp \right\}	\,,
\end{align}
where $\mbb{Z}^*_p=\{1,\ldots,p-1\}$. The left cosets of $BP$, i.e., the sets $\{F C_{\ag}:C_{\ag}\in BP\}$ for $F\in \SL{p}$, will be useful for our analysis. In particular, as we show in the following Lemma, the left coset representatives of $BP$, which we choose to be
\begin{align}\label{eq:representatives}
F_b &= \left(\begin{array}{cc}
  1 & b \\
 0 & 1 \\
\end{array}\right)	\quad (b\neq \infty)	\\
F_{\infty} &= \left(\begin{array}{cc}
  2 & 1 \\
 -1 & 0 \\
\end{array}\right)	\,,
\end{align}
can be used to label the single qudit MUBs. Moreover, while this labeling differs from that of the main text, we show that it preserves the $+1$ eigenstates.

\begin{lemma}\label{lem:prodpar}
The left cosets of $BP$ map the computational basis to the different mutually unbiased bases (MUBs) for a single qudit. Furthermore, 
\begin{align}\label{eq:repar}
\ket{\phi_1^0} &= U_{F_0} \ket{0}	\notag\\
\ket{\phi_{b+1}^0} &= U_{F_{b^{-1}}} \ket{0}	\quad	(b\neq 0,p)\notag\\
\ket{\phi_2^0} & = U_{F_{\infty}}\ket{0}\,,
\end{align}
where $\out{\phi_j^0} = \Pi_j^0$ is the labeling of stabilizer states in the main text.
\end{lemma}

\begin{proof}
To prove the first statement, note that each element of a left coset maps the computational basis to the same basis, since for any $F,F'$ in a left coset of $BP$ we can write $F' = FC_{\ag}$ for some $\ag$ (by definition of a left coset). Therefore
\begin{align}
\lvert\bra{k}U_F^{^\dagger} U_{F'}\ket{l}\rvert = \lvert\bra{k}U_{C_{\ag}}\ket{l} \rvert = \delta(k-\alpha l)	\,.
\end{align}
Since there are $p+1$ MUBs for a qudit of dimension $p$ and $p+1$ left cosets of $BP$ (since $\vert \SL{p}\vert /\vert BP\vert = p+1$), they must be in one-to-one correspondence.

The first line of Eq.~\eqref{eq:repar} is trivial since $F_0 = \mbb{I}$. To establish the second line, note that $\ket{\phi_b^0}$ is the unique state such that $D_{1,b-1}\ket{\phi_b^0}=\ket{\phi_b^0}$ for $b=1,\ldots,p+1$. Substituting $U_{F_{b^{-1}}}\ket{0}$ for $\ket{\phi_b^0}$, we have
\begin{align}
D_{1,b}U_{F_{b^{-1}}}\ket{0} &= U_{F_{b^{-1}}}U_{F_{b^{-1}}}^{\dagger}D_{1,b}U_{F_{b^{-1}}}\ket{0} \nonumber\\
&= U_{F_{b^{-1}}}D_{F_{b^{-1}}^{-1}1,b}\ket{0} \nonumber\\
&= U_{F_{b^{-1}}}D_{0,b}\ket{0} \nonumber\\
&= U_{F_{b^{-1}}}\ket{0}	\,,
\end{align}
and so $U_{F_{b^{-1}}}\ket{0} = \ket{\phi_b^0}$. The same argument shows that $\ket{\phi_2^0} = F_{\infty}\ket{0}$.
\end{proof}

\begin{lemma}\label{lem:entsep}
The set of entangled states that are not mutually unbiased with respect to a separable basis $(U_{F_b}\ket{k})\ket{l}$ is the set of entangled bases that are \Jam isomorphic to the left coset of $BP$ containing $F_b$.

Moreover, for entangled bases that are non-mutually unbiased to a given separable basis, 
\begin{align}\label{eq:ent_sep}
|\bra{k}U_{F_b}^\dag\bra{l} x,z,F_b C_{\ag}\rangle|^2 =\begin{cases} \tfrac{1}{p} \mbox{ if } x - bz =k-\alpha l \mbox{ and } b\neq\infty\\
\tfrac{1}{p} \mbox{ if } z =k-\alpha l \mbox{ and } b=\infty\\
0 \text{ otherwise,} \end{cases}
\end{align}
for all $b\in\mbb{Z}_{p,\infty}:=\mbb{Z}_p\cup\{\infty\}$, $k,l\in\mbb{Z}_p$ and $C_{\ag}\in BP$.
\end{lemma}

\begin{proof}
To prove the first statement, consider the inner product
\begin{align}
\vert(\bra{k}U_{F_b}^{\dagger})\bra{l}\ket{x,z,F}\vert^2 &= p^{-1} \vert\bra{k}U_{F_b}^\dag D_{x,z}U_F\ket{l}\vert^2 \nonumber\\
&= p^{-1} \vert \bra{k}D_{F_b^{-1} x,z}U_{F_b^{-1}F}\ket{l}\vert^2
\end{align}
for some fixed $F$. By Eq.~\eqref{eq:UF} the above inner product will be $p^{-2}$ unless $F_b^{-1} F = C_{\ag}$ for some $\ag$.

The second statement can be proven by straightforward calculation.
\end{proof}

Lemmas~\ref{lem:orth} and \ref{lem:entsep} establish the algebraic constraints within our parametrization that nonorthogonal projectors must satisfy, while Lemma~\ref{lem:prodpar} enables us to translate results from the parametrization used here to that of the main text. We now show how the algebraic constraints of Lemmas~\ref{lem:orth} and \ref{lem:entsep} can be used to show that for any $A^{\vec{r}}\in \mc{A}_{\rm sim}$, any set of $p^3 +1$ mutually nonorthogonal stabilizer projectors necessarily contains elements that are not in $\{\Pi\}^{\vec{r}}$. Consequently, $\alpha(\Gamma^{\vec{r}}) < p^3 + 1$ for any $A^{\vec{r}}\in \mc{A}_{\rm sim}$.

\begin{theorem}\label{thm:upper}
For any $A^{\vec{r}}\in \mc{A}_{\rm sim}$, the independence number of the exclusivity graph $\Gamma^{\vec{r}}$ satisfies $\alpha(\Gamma^{\vec{r}}) < p^3 + 1$.
\end{theorem}

\begin{proof}
Let $A^{\vec{r}}\in \mc{A}_{\rm sim}$ and assume, for the purpose of obtaining a contradiction, that there exists a set $I$ of $p^3+1$ mutually nonorthogonal elements of $\{\Pi\}^{\vec{r}}$. Since $\{\Pi\}^{\vec{r}}$ contains elements from $p^3 + 1$ orthonormal bases, such a set must contain one element from each basis.

Therefore we can parametrize $I = \sep \cup \ent$ by
\begin{align}\label{eq:parameterization}
\sep &= \{(U_{F_b}\otimes 1)\ket{k_b l_b}:b\in\mbb{Z}_{p,\infty}\} \nonumber\\
\ent &= \{\ket{x^b_{\ag},z^b_{\ag},F_b C_{\ag}}:b\in\mbb{Z}_{p,\infty}\}
\end{align}
for some choice of $\{k_b,l_b,x^b_{\ag},z^b_{\ag}:b\in\mbb{Z}_{p,\infty},\alpha \in\mbb{Z}_p^*,\gamma\in\mbb{Z}_p\}$. In order for $\sep$ to be a set of mutually nonorthogonal states, $l_b =: l$ for all $b$. Note that we abuse notation slightly by referring to $I$ as a set of pure states rather than the corresponding projectors.

Without loss of generality we set $\vec{r} = 0$ since any $A^{\vec{r}}\in \mathcal{A}_{\textsc{SIM}}$ can be written as $D_{x,z}A^{0} D_{x,z}^{\dagger}$ for some $x$ and $z$ $^{40}$ and applying a unitary to $\{\Pi\}^{\vec{r}}$ does not change the exclusivity graph (since orthogonality is preserved by unitaries).

Therefore to prove that $I$ can be contained in $\{\Pi\}^{\vec{r}}$ for any $A^{\vec{r}}\in \mc{A}_{\rm sim}$, it is sufficient to prove that $k_b = 0$ for some $b$, since the corresponding projector is not in $\{\Pi\}^{0}$ by Lemma~\ref{lem:prodpar}. We will prove this by iteratively using the constraints from Lemmas~\ref{lem:orth} and \ref{lem:entsep}.

By Lemma~\ref{lem:orth}, requiring
\begin{align}
\inner{x^0_{\ag},z^0_{\ag},C_{\ag}}{x^b_{\ag},z^b_{\ag},F_b C_{\ag}} \neq 0
\end{align}
is equivalent to requiring
\begin{align}\label{eq:zag}
z^b_{\ag} &= z^0_{\ag}\ \forall \ag \quad (b\neq \infty) \nonumber\\
x^{\infty}_{\ag} &= x^0_{\ag} + z^0_{\ag} - z^{\infty}_{\ag} \ \forall \ag \quad (b= \infty)\,,
\end{align}
where we have used
\begin{align}
C_{\ag}^{-1} F_b C_{\ag} = \begin{cases}
\mat{cc}{1 + \alpha^{-1}\gamma b & \alpha^{-2} b \\ -\gamma^2 b & 1 - \alpha^{-1}\gamma b} & (b\neq \infty)\\
\mat{cc}{2 + \alpha^{-1}\gamma & \alpha^{-2} \\ -(\alpha+ \gamma)^2 & - \alpha^{-1}\gamma} & (b= \infty).\\
\end{cases}
\end{align}

We now consider the requirement that the sets $\sep$ and $\ent$ are pairwise mutually nonorthogonal. By Eq.~\eqref{eq:ent_sep},
\begin{align}
\bra{k_b l_b}(U_{F_b}^\dag\otimes 1)\ket{x^b_{\ag},z^0_{\ag}, F_b C_{\ag}} \neq 0
\end{align}
is equivalent to
\begin{align}\label{eq:xag}
x^b_{\ag} &= b z^0_{\ag} + k_b - l \alpha \ \forall \ag \quad (b\neq \infty) \nonumber\\
z^{\infty}_{\ag} &= l\alpha - k_{\infty} =: z^{\infty}_{\alpha} \ \forall \ag \quad (b= \infty)\,,
\end{align}
which completely characterizes the restrictions on $\ent$ such that it contains no elements orthogonal to any element of a fixed $\sep$.

This then completely specifies every parameter except $z^0_{\ag}$ and $\{k_b:b\in\mbb{Z}_{p,\infty}\}$. To specify the remaining parameters, we will have to impose further constraints on the elements of $\ent$ to ensure $\ent$ contains no pairs of mutually orthogonal elements.

To do this, note that for all $b\in\mbb{Z}_p$ and $c\in\mbb{Z}_p^*$, $F_b F_c = F_{b+c}$ and $\Tr F_{c}C_{2,-(2c)^{-1}}=2$. Therefore, by Lemma~\ref{lem:orth},
\begin{align}\label{eq:crossterms}
\inner{x^b_{1,\gamma},z^0_{1,\gamma},F_b C_{1,\gamma}}{x^{b+c}_{2,\gamma'},z^0_{2,\gamma'},F_{b+c} C_{2,\gamma'}}\neq 0 \,,
\end{align}
where $\gamma' = 2^{-1}\gamma - (2c)^{-1}$, is equivalent to
\begin{align}\label{eq:kbc}
k_{b+c} - k_b = l - c(2z^0_{2,\gamma'}-z^0_{1,\gamma})
\end{align}
for all $b,\gamma\in\mbb{Z}_p$ and $c\in\mbb{Z}_p^*$, where we have used Eq.~\eqref{eq:zag} and \eqref{eq:xag}. If the right-hand-side of Eq.~\eqref{eq:kbc} is nonzero for any value of $c$ or $\gamma$, then $\{k_b:b\in\mbb{Z}_p\}\equiv \mbb{Z}_p$ and so $k_{b'}=0$ for some value of $b'$, which would complete the proof.

Therefore the only way $I\subset\{\Pi\}^{0}$ is if the right-hand-side of Eq.~\eqref{eq:kbc} is zero for all $c\in\mbb{Z}_p^*$ and $\gamma$, which implies $k_b = k \neq 0$ for all $b\in\mbb{Z}_p$ and
\begin{align}
l = c(2z^0_{2,\gamma'}-z^0_{1,\gamma})	\ \forall c\neq 0,\gamma \,.
\end{align}
This can be solved to obtain
\begin{align}
z^0_{1,\gamma} &= z^0_{1,0} - \gamma l	\nonumber\\
z^0_{2,\gamma} &= z^0_{2,0} - \gamma l	
\end{align}

We now consider the case where $b=\infty$ in Eq.~\eqref{eq:crossterms}, which implies that
\begin{align}
\inner{x^{\infty}_{1,0},z^{\infty}_{1},F_{\infty}}{x^c_{2,\gamma'},z^0_{2,\gamma'},F_c C_{2,\gamma'}}\neq 0 \,,
\end{align}
for $c\neq -2$ where $\gamma'=2^{-1}c - 1$ is equivalent to
\begin{align}
(c-1)k_{\infty} = 2z^0_{2,0}-z^{0}_{1,0}	\,,
\end{align}
which can only hold independently of $c$ if $k_{\infty} = 0$.
\end{proof}

%


%
%
%
%

\end{document}